\documentclass[12pt,draftcls,journal,onecolumn]{IEEEtran}
\usepackage{amsfonts,color,morefloats,pslatex}
\usepackage{amssymb,amsthm, amsmath,latexsym}

\newtheorem{theorem}{Theorem}
\newtheorem{lemma}[theorem]{Lemma}

\newtheorem{corollary}[theorem]{Corollary}

\newtheorem{definition}[theorem]{Definition}
\newtheorem{example}[theorem]{Example}

\newtheorem{remark}[theorem]{Remark}
\newtheorem{conjecture}[theorem]{Conjecture}

\newcommand{\bF}{ {\mathbb F}}

\newcommand{\gf}{{\mathrm{GF}}}

\usepackage{blindtext}

\begin{document}

\title{Several classes of PcN power functions over finite fields \footnote{}}

\author{Xiaoqiang Wang,\,\,\, Dabin Zheng*{\thanks{*Corresponding author. This work was partially supported by the National Natural Science Foundation of China under Grant Numbers 12001175 and 11971156.
\newline \indent ~~Xiaoqiang Wang and Dabin Zheng are with the Hubei Key Laboratory of Applied Mathematics, Faculty of Mathematics and Statistics, Hubei University, Wuhan 430062, China (E-mail: waxiqq@163.com, dzheng@hubu.edu.cn)
}}}

\date{
}

\maketitle

\begin{abstract}
Recently, a new concept called multiplicative differential cryptanalysis and the corresponding $c$-differential uniformity were introduced by Ellingsen et al.~\cite{Ellingsen2020}, and then some low differential uniformity functions were constructed. In this paper, we further study the constructions of perfect $c$-nonlinear (PcN) power functions. First, we give a necessary and sufficient condition for the Gold function to be PcN and a conjecture on all power functions to be PcN over $\gf(2^m)$. Second, several classes of PcN power functions are obtained over finite fields of odd characteristic for $c=-1$ and our theorems generalize some results in~\cite{Bartoli,Hasan,Zha2020}. Finally, the $c$-differential spectrum of a class of almost perfect $c$-nonlinear (APcN) power functions is determined.
\end{abstract}
\vskip 6pt
\noindent {\it Keywords:} {\small C-differential uniformity, perfect c-nonlinear function, almost perfect c-nonlinear function, differential spectrum}

\noindent {\it \small 2010 Mathematics Subject Classification:} {\small 94A60, 11T71.}
\vskip 35pt

\leftskip 0.0in
\rightskip 0.0in

\section{Introduction}

Differential cryptanalysis proposed by Biham and Shamir in literature [5] is a powerful analysis method to attack block cipher, which has attracted extensive attention of researchers.  The basic idea of differential cryptanalysis is to recover the key values with the greatest possibility by analyzing the influence of a specific plaintext difference on the ciphertext difference. The security of cryptographic functions against differential attacks has been extensively studied in the past 30 years. In order to measure the ability of a given function to resist differential attack, Nyberg introduced the concept of differential uniformity in~\cite{Nyberg1994} :
Let $\gf(p^m)$ denote the finite field with $q$ elements. A function $f$ from $\gf(p^m)$  to itself
is called differentially $\Delta_f$-uniform, where
\begin{equation*}
\Delta_f=\underset{0\neq a \in \gf(p^m)}{\rm max}\underset{b \in \gf(p^m)}{\rm max}\,|\,\{ x \in \gf(p^m) \, | \, f(x+a)-f(x)=b \} |.
\end{equation*}
The lower the quantity of $\Delta_f$, the stronger the ability of the function $f(x)$ to resist differential attack. If $\Delta_f= 1$ and $\Delta_f= 2$, then $f$ is called a perfect nonlinear (PN) function and an almost perfect nonlinear (APN) function, respectively. In the past many years, a lot of progress on the constructions of PN and APN functions have been made. The reader is referred to~\cite{Coulter97,Dembowski68,Ding06,Dobbertin99,Dobbertin991,zha091,zha09} and the references therein for information.

Recently, a new type of differential was proposed in \cite{Borisov2002}. The authors extended the type of differential cryptanalysis by using modular multiplication as a primitive operation. For a vectorial
Boolean function $f$, they argued that one should look at new type of differential $(f(cx), f(x))$ and not only $(f(x + a), f(x))$. Based on this work, Ellingsen et al. in \cite{Ellingsen2020} defined a new concept called {\it multiplicative differential}, and proposed the corresponding concept of $c$-differential uniformity as follows.
\begin{definition}\label{definition}
Let $\gf(p^m)$ denote the finite field with $p^m$ elements and $a, c\in \gf(p^m)$. For a function $F(x)$ from $\gf(p^m)$ to itself, the (multiplicative) $c$-derivative of $F(x)$ with respect to $a$ is defined as
\begin{equation*}
{}_cD_aF(x)=F(x+a)-cF(x), \,\,{\rm for}\,\, {\rm all}\,\,\, x.
\end{equation*}
For $a, b\in \gf(p^m)$, let ${}_c\Delta_F(a,b)=\# \{x\in \gf(p^m):F(x+a)-cF(x)=b \}$. We call ${}_c\Delta_F=max$ $ \{ {}_c\Delta_F(a,b):a,b\in \gf(p^m), $ and $a \neq 0$ if $c = 1\}$ the $c$-differential uniformity of $F(x)$.
\end{definition}

If ${}_c\Delta_F=\delta$, then we say that $F$ is differentially $(c, \delta)$-uniform. If $\delta = 1$ and $\delta=2$, then $F$ is called a perfect $c$-nonlinear (PcN) function and an almost perfect c-nonlinear (APcN) function, respectively. If $c = 1$, then the $c$-differential uniformity becomes the usual differential uniformity, and PcN and APcN functions become PN and APN functions, respectively. It is known that APN functions over finite fields of even characteristic have the lowest differential uniformity. However, for the $c$-differential uniformity, there exist PcN functions.

Since the power functions with low differential uniformity are an ideal choice for S-box design, these functions have attracted a lot of attention, especially the PcN and APcN power functions. The reader is referred to ~\cite{Bartoli,Borisov2002,Mesnager2020,Hasan,Yanar,Yan2020,Zha2020} and the references therein for information. For convenience, we list the known PcN and APcN power functions in Table~1. Among other results,
the references~\cite{Bartoliar,Wuar} also studied PcN and APcN multinomials. Table~1 shows that there are very few results on PcN power functions. For the case over finite fields with even characteristic, except for some very special cases, the Gold function is the only known PcN power function. For the case over finite fields with odd characteristic, most known PcN monomials $x^d$ are either over finite fields $\gf(3^m), \gf(5^m)$ for any positive integer $m$, or over small extensions of any odd prime field $\gf(p)$. These exponents $d$ can be seen as special solutions of $d({p^k+1})\equiv2 \pmod {p^m-1}$.

In this paper, our main objective is to construct some infinite classes of PcN power functions.
First, we give a necessary and sufficient condition for the Gold function to be PcN and  a conjecture of necessity and sufficiency conditions for all power functions to be PcN over finite fields with even characteristic. Second,  several classes of PcN power functions $x^d$ over $\gf(p^m)$ with $c=-1$ are proposed, where $p$ is an odd prime and $d$ satisfies $d({p^k+1})\equiv2 \pmod {p^m-1}$. Some known PcN power functions in~\cite{Bartoli,Hasan,Zha2020} are some special cases of our results. Finally, the $c$-differential spectrum of a class of APcN power functions is given.

The rest of this paper is organized as follows. Section \ref{sec-2} documents some preliminaries. Section \ref{sec-3} gives the necessity and sufficiency for the Gold function being PcN and a conjecture for all power functions being PcN over finite fields with even characteristic. Section~\ref{sec-4} obtains some PcN power functions over finite fields with odd characteristic. Moreover, the $c$-differential spectrum of a class of APcN power functions
is given. Section~\ref{sec:concluding} concludes this paper.

\begin{table}[h]
{\caption{\rm PcN and APcN Power functions $F(x)=x^d$ over $\gf(p^m)$ with $c\neq 1$}\label{table11}
\begin{center}
\vspace{-7mm}
\begin{tabular}{|c|c|c|c|c|}\hline
$p$     &$d$                                  &condition                                                                       &${}_c\Delta_F$  & Refs.\\\hline
any    &$2$                                  &$c \neq 1$                                                                  &2      &\cite{Ellingsen2020}    \\\hline
any    &$p^m-2$                              &$c=0$                                                                       &1         &\cite{Ellingsen2020}    \\\hline
2       &$2^m-2$                              &$c \neq 0$, ${\rm Tr}^m_1(c)={\rm Tr}^m_1(c^{-1})=1$                                        &2          &\cite{Ellingsen2020}    \\\hline
odd  &$p^m-2$                              &$\begin{array}{c}c=4,4^{-1}\,\, or\\
\chi(c^2-4c)=\chi(1-4c)=-1\end{array}$                            &2         &\cite{Ellingsen2020}    \\\hline
3       &$\frac{3^k+1}{2}$                    &$c=-1$, $\frac{n}{\gcd(k,m)}=1$                                             &1          &\cite{Ellingsen2020}    \\\hline
odd  &$\frac{p^2+1}{2}$                    &$c=-1$, $m$ odd                                                           &1       & \cite{Bartoli}     \\\hline
odd  &$p^2-p+1$                            &$c=-1$, $m=3$                                                               &1   & \cite{Bartoli}       \\\hline
odd  &$\begin{array}{c}p^4+(p-2)p^2\\
+(p-1)p+1\end{array}$                            &$c=-1$, $m=5$                                                               &1          & \cite{Hasan}    \\\hline
odd  &$(p^5+1)/(p+1)$                            &$c=-1$, $m=5$                                                               &1      & \cite{Hasan}    \\\hline
odd  &$\begin{array}{c}(p-1)p^6+p^5+(p-2)p^3\\
+(p-1)p^2+p\end{array}$                            &$c=-1$, $m=7$                                                               &1          & \cite{Hasan}    \\\hline
odd  &$\begin{array}{c}(p-2)p^6+(p-2)p^5+\\
+(p-1)p^4+p^3+p^2+p\end{array}$                            &$c=-1$, $m=7$                                                               &1          & \cite{Hasan}    \\\hline
odd  &$(p^7+1)/(p+1)$                            &$c=-1$, $m=7$                                                               &1          & \cite{Hasan}    \\\hline
3       &$\frac{3^n+3}{2}$                              &$c=-1$, $m$ even                                                  &2        &\cite{Mesnager2020}  \\\hline
3    &$3^n-3$                              &$c=0$                                      &$2$ &\cite{Mesnager2020}  \\\hline
odd  &$\frac{p^k+1}{2}$                    &$v_2(m)\leq v_2(k)+1$, $c=-1$                                         &1       &\cite{Mesnager2020}  \\\hline
odd  &$p^k+1$                              &$v_2(m)\leq v_2(k)$, $1 \neq c \in  \mathbb{F}_{p^{\gcd(m,k)}}$                                   &2          &\cite{Mesnager2020}  \\\hline
2  &$2^k+1$                              &$v_2(m)\leq v_2(k)$,  $k\geq 2$, $1 \neq c \in  \mathbb{F}_{2^{\gcd(m,k)}}$                                   &1         &\cite{Mesnager2020}  \\\hline
3       &$\frac{3^k+1}{2}$                    &$k$ odd, $\gcd(k,m)=1$, $c=-1$                                            &2        &\cite{Yan2020}   \\\hline
3       & $\begin{array}{c}\frac{3^k+1}{2}d\equiv \frac{3^m+1}{2} \pmod {3^m-1}\\
                          \text{$d$ odd}  \end{array}$
                          & $\begin{array}{c}\,\, \text{$k$ and $m$ are odd } \\
\text{such that}\,\, {\rm\gcd}(m,k)=1\end{array}$   &$1$    &\cite{Zha2020}      \\\hline
5          & $\begin{array}{c}\frac{5^k+1}{2}d\equiv \frac{5^m+1}{2} \pmod {5^m-1}\\
                          \text{$d$ odd}  \end{array}$                         &$\begin{array}{c}\,\, \text{$k$ and $m$ are positive} \\ \text{ integer
such that}\,\, {\rm\gcd}(2m,k)=1\end{array}$                &$1$    &\cite{Zha2020}    \\\hline
\end{tabular}
\end{center}}
\end{table}

\section{Notation and Preliminaries}\label{sec-2}
Throughout this paper, we always let $m,k,d$ be positive integers. Let $v_2(\cdot)$ be the 2-adic order function and $v_2(0)=\infty$. Let $\gf(p^m)$ denote the finite field with $p^m$ elements, and $\gf(p^m)^*$ the set of non-zero elements in $\gf(p^m)$. Let $\eta$ be the quadratic character of $\gf(p^m)^*$, i.e., $\eta(x)=x^{\frac{p^m-1}{2}}$ for $x\in \gf(p^m)^*$. Then $\eta(x)=1$ if $x$ is a square element in $\gf(p^m)^*$ and  $\eta(x)=-1$  if $x$ is a  non-square element in $\gf(p^m)^*$.

Let $F(x)$ be a power function over $\gf(p^m)$. It is easy to check that ${}_c\Delta_F(a,b)={}_c\Delta_F(1,b)$ for $a\in \gf(p^m)^*$ and  ${}_c\Delta_F(a,b)={\rm gcd}(d,p^m-1)$ for $a=0$ and $c\neq 1$. Hence, the following result on the $c$-differential uniformity of power functions is easily obtained, which was first given in~\cite{Yan2020}.

\begin{lemma}\cite[Lemma 1]{Yan2020}\label{lemma1}
Let $F(x)=x^d$ be a power function over $\gf(p^m)$. Then
\begin{equation*}
{}_c\Delta_F={\rm max}\left\{ \{{}_c\Delta_F(1,b):b\in\gf(p^m)\}\cup\{{\rm gcd}(d,p^m-1)\} \right\}.
\end{equation*}
\end{lemma}

PcN functions have the lowest $c$-differential uniformity and have been widely studied. The following result on PcN power functions is well-known and has been analyzed in \cite{Mesnager2020,Riera991,Hasan,Yan2020}.

\begin{lemma}\cite[Theorem 6]{Mesnager2020}\label{lem:pk+1div2}
Let $p$ be an odd prime and $m, k$ be integers with $1\leq k< m$ and $m\geq 3$. Let $F(x)=x^{\frac{p^k+1}{2}}\in \gf(p^m)[x]$. If $c=-1$, then $F$ is PcN if and only if $v_2(m)\leq v_2(k)+1$. Otherwise,  ${}_{-1}\Delta_{F}=\frac{p^{\gcd(k,m)}+1}{2}$.
\end{lemma}

Following the definition in \cite{Blondeau2010}, the $c$-differential spectrum of a function is given as follows.
\begin{definition}\label{definition1}
Let $F(x)=x^d$ be a function over $\gf(p^m)$. Denote by $\omega_i$ the number of output differences $b$ that occur $i$ times, that is,
$\omega_i=\#\{ b \in \gf(p^m)\, | \, {}_c\Delta_F(a,b)=i\}$
for each $0\leq i\leq {}_c\Delta_F$. The differential spectrum of $F$ is defined to be the set
$$\mathbb{S}=\{\omega_i\,|\, 0\leq i \leq {}_c\Delta_F(a,b) \,\, and \,\, \omega_i>0\}.$$
\end{definition}

The following lemma will be used to compute the $c$-differential spectrum of some APcN functions.

\begin{lemma}\cite[Theorem 5.6]{Bluher2020}\label{lemma4}
Let $g(x)=x^{p^k+1}-bx+b$ with $b \in \gf(p^m)^*$. Then the number of the solutions to $g(x)=0$ in $\gf(p^m)$ is $0$, $1$, $2$ or $p^{\gcd(m,k)}+1$. Let $N_i$ denote the number of $b\in \gf(p^m)^*$ such that $g(x)=0$ has exactly $i$ roots in $\gf(p^m)$. Let $Q=p^{\gcd(m,k)}$ and $h=[\gf(p^m): \gf(p^{\gcd(m,k)})]$, then the following statements hold.
\begin{enumerate}
\item[{\rm (1)}] If $h$ is even, then
$$N_0=\frac{Q^{h+1}-Q}{2(Q+1)},\,\,N_1=Q^{h-1},\,\,N_2=\frac{(Q-2)(Q^h-1)}{2(Q-1)},\,\, N_{Q+1}=\frac{Q^{h-1}-Q}{Q^2-1}.$$
\item[{\rm (2)}] If $p$ and $h$ are odd, then
$$N_0=\frac{Q^{h+1}-1}{2(Q+1)},\,\,N_1=Q^{h-1},\,\,N_2=\frac{Q^{h+1}-2Q^h-2Q+3}{2(Q-1)},\,\, N_{Q+1}=\frac{Q^{h-1}-Q}{Q^2-1}.$$
\item[{\rm (3)}] If $p$ is even and $h$ is odd, then
$$N_0=\frac{Q^{h+1}+Q}{2(Q+1)},\,\,N_1=Q^{h-1}-1,\,\,N_2=\frac{(Q-2)(Q^h-1)}{2(Q-1)},\,\, N_{Q+1}=\frac{Q^{h-1}-1}{Q^2-1}.$$
\end{enumerate}
\end{lemma}

The following is a known result, which will be used throughout this paper.
\begin{lemma}\label{lemma5}
Let $p$ be a prime and $m,k$ be positive integers, then
\[\gcd(p^k+1,p^m-1)=\left\{ \begin{array}{lll}
           {\frac{2^{\gcd(2k,m)}-1}{2^{\gcd(k,m)}-1}}, & {\rm if}\, \,\, p = 2, \vspace{2mm}\\
           {2}, & {\rm if}\, \,\, v_2(m)\leq v_2(k), \vspace{2mm}\\
           {p^{\gcd(k,m)}+1}, & {\rm if}\, \,\, v_2(m)>v_2(k). \end{array}  \right.\]
\end{lemma}

In order to discuss the existence of the solutions of a congruence equation, we need the following known fact.
\begin{lemma}\label{lem-05}
Let $\phi$, $\varphi$, $\mu$ be three non-zero elements in $\gf(p^m)$. Then the congruence equation $\phi x\equiv \varphi \pmod \mu$ has solutions if and only if ${\rm gcd}(\phi,\mu)\,|\,\varphi$.
\end{lemma}

\section{PcN power functions over $\gf(2^m)$} \label{sec-3}

In this section, we present a necessary and sufficient condition for the Gold function to be PcN and give a conjecture of necessity and sufficiency conditions for all power functions to be PcN. To this end, we first give a general result on PcN monomials over $\gf(p^m)$, where $p$ is a prime.

\begin{lemma}\label{lem:inversepcn}
Let $F(x)=x^d$ be a PcN function over $\gf(p^m)$, then $F'(x)=x^{d^{-1}}$ is also a P$c'$N function, where $c'=c^d$ and $d^{-1}$ is the inverse of $d$ modulo $p^m-1$. Moreover, $c=c'$ if $c=\pm 1$ or  $0$.
\end{lemma}
\begin{proof}
 If $c=0$, it is easy to see that the result holds. In the following, we always assume that $c\neq 0$. By the definition of PcN functions, for any $a,b\in \gf(p^m)$,
\begin{equation*}
\begin{split}
(x+a)^d-cx^d=b
 \end{split}
\end{equation*}
has only one solution in $\gf(p^m)$. If $b=0$, then the above equation becomes $(x+a)^d=cx^d$, which implies $(1+a/x)^d=c$ has only one solution for any $a \in \gf(p^m)$. Hence, $(1+a/x)^d$ is a permutation polynomial over $\gf(p^m)$. This means that $\gcd(p^m-1,d)=1$.
Then $d$ has the inverse modulo $p^m-1$ and $c=c^d$ if $c=\pm 1$. Hence,
\begin{equation}\label{eq1ssd}
\begin{split}
(x+a)^d-cx^d=b&\Longleftrightarrow (x+a)^d=cx^d+b\Longleftrightarrow (x+a)=(cx^d+b)^{d^{-1}}.
\end{split}
\end{equation}
Let $cx^d=y$, then $x$ can be expressed as $x=(yc^{-1})^{d^{-1}}$ and the equation in (\ref{eq1ssd}) becomes $(y+b)^{d^{-1}}-c^dy^{d^{-1}}=a$. Hence, $x^{d^{-1}}$ is a P$c'$N function if $x^d$ is a PcN function over $\gf(p^m)$, where $c'=c^d$. This completes the proof.
\end{proof}

It is known that there is no PN functions, but exist PcN functions over finite fields of even characteristic.
In \cite{Mesnager2020,Riera991,St2020,Yan2020}, the authors considered the c-differential uniformity of the Gold function $F(x)=x^{2^k+1}$ over $\gf(2^m)$ and showed that
the Gold function has low $c$-differential uniformity if $c$, $k$ and $m$ satisfy some conditions. In the following theorem, we continue to analysis the Gold function $F(x)=x^{2^k+1}$ and give a necessary and sufficient condition for the Gold function to be PcN.

\begin{theorem}\label{thm:PcNmonomials}
Let $F(x)=x^{2^k+1}$ over $\gf(2^m)$. Then $F(x)$ is PcN if and only if  $v_2(m)\leq v_2(k)$ and $c\in \gf(2^{\gcd(k, m)})\backslash \{1\}$.
\end{theorem}
\begin{proof}
 It is known there does not exist PcN functions over $\bF_{2^m}$ if $c=1$. From Lemma \ref{lemma5},  $F(x)$ is PcN if and only if  $v_2(m)\leq v_2(k)$ if $c=0$.  In the following, we always assume that $c\neq 0$ and $c\neq 1$.

Assume that $F(x)$ is PcN, then $\Delta(x)=b$ has only one solution for $b \in \gf(2^m)$, where
\begin{equation*}
\begin{split}
\Delta(x)&=(x+1)^{2^k+1}+cx^{2^k+1}=(c+1)x^{2^k+1}+x^{2^k}+x+1.
\end{split}
\end{equation*}
Since gcd$(2^k,2^m-1)=1$, there exists an element $\beta \in \gf(2^m)^*$ such that $\beta^{2^k}=\frac{1}{c+1}$. Let $y=x+\beta$, then $\Delta(x)=b$ can be rewritten as
 \begin{equation}\label{eq:fas}
\begin{split}
b&=(c+1)(y+\beta)^{2^k+1}+(y+\beta)^{2^k}+y+\beta+1\\
&=(c+1)y^{2^k+1}+(\beta^{2^k}(c+1)+1)y+(\beta(c+1)+1)y^{2^k}+\beta^{2^k+1}+\beta^{2^k}+\beta+1\\
&=(c+1)y^{2^k+1}+(\beta^{1-2^k}+1)y^{2^k}+\beta^{2^k+1}+\beta^{2^k}+\beta+1\\
&=((c+1)y+\beta^{1-2^k}+1)y^{2^k}+\beta^{2^k+1}+\beta^{2^k}+\beta+1.
\end{split}
\end{equation}
Let $b=\beta^{2^k+1}+\beta^{2^k}+\beta+1$, then Eq. (\ref{eq:fas}) becomes
$$((c+1)y+\beta^{1-2^k}+1)y^{2^k}=0$$
and this equation has only one solution $y=0$ since $\Delta(x)=b$ has only one solution for $b \in \gf(2^m)$. This means that $(c+1)y+\beta^{1-2^k}+1=0$ has not solutions except for $y=0$. Since $(c+1)y$ is a permutation polynomial over $\gf(2^m)$, then $\beta^{1-2^k}+1=0$.
Hence, Eq. (\ref{eq:fas}) can be rewritten as
\begin{equation}\label{eqccc1}
(c+1)y^{2^k+1}+\beta^{2^k+1}+\beta^{2^k}+\beta+1=b.
 \end{equation}
By the definition of PcN, we can deduce $\gcd(2^m-1, 2^k+1)=1$. From Lemma \ref{lemma5}, we have $v_2(m)\leq v_2(k)$.
From $\beta^{2^k}=\frac{1}{c+1}$, we have $\beta=\frac{1}{c^{2^{m-k}}+1}$. Then $\beta^{1-2^k}+1=0$ if and only if
 $$\frac{\beta}{\beta^{2^k}}=\frac{1}{c^{2^{m-k}}+1}/\frac{1}{c+1}=1, \,\,i.e.,\,\, c^{2^{m-k}-1}=1.$$
Since $\gcd(2^{m-k}-1,2^m-1)=\gcd(2^k-1,2^m-1)$, we have $c\in \gf(2^{\gcd(m,k)})$. Hence, we deduce that
$v_2(m)\leq v_2(k)$ and $c\in \gf(2^{\gcd(m,k)})\setminus\{1\}$ if $F(x)$ is PcN.

Now, we assume that $v_2(m)\leq v_2(k)$ and $c\in \gf(2^{\gcd(m,k)})\setminus\{1\}$. From Lemma \ref{lemma5},
\begin{equation}\label{eq:df}
(c+1)y^{2^k+1}+\beta^{2^k+1}+\beta^{2^k}+\beta+1
\end{equation}
is a permutation polynomial over $\gf(2^m)$, where $\beta=\frac{1}{c^{2^{m-k}}+1}$. Let $y=x+\beta$. From Eq.(\ref{eq:fas}) we know that the polynomial in (\ref{eq:df}) can be rewritten as $(c+1)x^{2^k+1}+x^{2^k}+x+1$, which is also a permutation polynomial. This means that that $F(x)$ is PcN.
\end{proof}

\begin{remark}
In \cite[Theorem 4]{Mesnager2020}, the authors proposed the following result: Let $2\leq k< m$, $m\geq 3$ and $F(x)=x^{2^k+1}$ be the Gold function over $\gf(2^m)$. Assume that $m=ld$, where $d=\gcd(m,k)$ and $l\geq 3$. If $1\neq c \in \gf(2^d)$, the $c$-differential uniformity of $F$ is ${}_c\Delta_F=\frac{2^{\gcd(2k,m)}-1}{2^{\gcd(k,m)}-1}$. If $c \in \gf(2^m)\setminus \gf(2^d)$, the $c$-differential uniformity of $F$ is ${}_c\Delta_F=2^d+1$. From this result, it is easy to get that when $1\neq c \in \gf(2^d)$, $F(x)=x^{2^k+1}$ is PcN if $2\leq k< m$ and $v_2(m)\leq v_2(k)$,  where $m=ld$, $d=\gcd(m,k)$ and $l\geq 3$. However, it cannot get the the necessity and sufficiency for the Gold function $F(x)=x^{2^k+1}$ to be PcN for any $k$.
\end{remark}

Let $d=2^j$ for $0\leq j \leq m-1$ and $c\in \gf(2^m)\backslash \{1\}$,  one can easily deduce that the equation
\[ (x+a)^d + cx^d =b \]
has only one solution in $\gf(2^m)$ for any $a, b\in \gf(2^m)$. Moreover, let $c\in \gf(2^m)$, the $c$-differential uniformity of the power functions $x^d$ and $x^{dp^h}$ is the same for any non-negative integer $h$. Then combining Lemmas~\ref{lem:inversepcn} and~\ref{thm:PcNmonomials}, we have the following result.

\begin{corollary}\label{cor:PcNmonomials}
Let $F(x)=x^d$ be a monomial over $\gf(2^m)$. Then $F(x)$ is a PcN function if one of the following conditions hold:
\begin{enumerate}
\item[{\rm (1)}] $d=2^j$ for $0\leq j \leq m-1$ and $c\in \gf(2^m)\backslash \{1\}$.
\item[{\rm (2)}] $d$ belongs to $\left\{ 2^j(2^k+1), \,\, j=0,1, \cdots, m-1\right\}$ or the set of their multiplicative inverses modulo $(2^m-1)$ for some positive integer $k$ with $v_2(m)\leq v_2(k)$ and
$c\in \gf(2^{\gcd(k, m)})\backslash \{1\}$.
\end{enumerate}
\end{corollary}

\begin{example}
Let $m=6$ and $d\in U$, where
$$U=\{1,2,4,8,10,13,16,17,19,20,26,32,34,38,40,41,52\}.$$
Then $F(x)=x^d$ is PcN when $c$ satisfies the corresponding condition in Corollary \ref{cor:PcNmonomials}. These results have been verified by Magma programs.
\end{example}

We checked that the necessity of Corollary~\ref{cor:PcNmonomials} by numerical experiment and found that the necessity of Corollary~\ref{cor:PcNmonomials} is also right for $2\leq m\leq 10$. However, it is not clear that whether the necessity of Corollary~\ref{cor:PcNmonomials} holds for any $m$. So, we give the following conjecture.

\begin{conjecture}
Let $F(x)=x^d$ be a power function over $\gf(2^m)$, then $F(x)=x^d$ is a PcN function if and only if one of the following conditions holds:
\begin{enumerate}
\item[(1)] $d=2^j$ for $0\leq j \leq m-1$ and $c\in \gf(2^m)\backslash \{1\}$.
\item[(2)] $d$ belongs to $\left\{ 2^j(2^k+1), \,\, j=0,1, \cdots, m-1\right\}$ or the set of their multiplicative inverses modulo $(2^m-1)$ for some positive integer $k$ with $v_2(m)\leq v_2(k)$ and
$c\in \gf(2^{\gcd(k, m)})\backslash \{1\}$.
\end{enumerate}
\end{conjecture}

%

\section{PcN and APcN power functions over $\gf(p^m)$}\label{sec-4}

In this section, let $p$ be an odd prime. We will study the $c$-differential uniformity of some monomials and obtain some PcN power functions over $\gf(p^m)$ for $c=-1$, and generalize some results in \cite{Bartoli,Hasan,Zha2020}. To this end, we need to investigate the solutions of the following equations.
\begin{equation}\label{system1}
\begin{split}
&\mathbf{(I)}\left\{ \begin{array}{ll}
           {x_1^2+y_1^2=1},\\
           x_1^{p^k+1}-y_1^{p^k+1}=-b^{\frac{p^k+1}{2}},\end{array}  \right.\,\,\,\,\,\,
           \mathbf{(II)}\left\{ \begin{array}{ll}
           {x_2^2-y_2^2=1},\\
           {x_2^{p^k+1}+y_2^{p^k+1}=-b^{\frac{p^k+1}{2}}},\end{array}  \right. \\  \\
&  \mathbf{(III)}\left\{\begin{array}{ll}
           {x_3^2-y_3^2=-1},\\
           {x_3^{p^k+1}+y_3^{p^k+1}=b^{\frac{p^k+1}{2}}},\end{array}  \right.\,\,\,\,\,\,
           \mathbf{(IV)}\left\{ \begin{array}{ll}
           {x_4^2+y_4^2=-1},\\
           {x_4^{p^k+1}-y_4^{p^k+1}=b^{\frac{p^k+1}{2}}}\end{array}.  \right.
\end{split}
\end{equation}

\begin{lemma}\label{lem:ddf}
Let $p^k\equiv 3 \pmod 4$  and $p^m\equiv 3 \pmod 4$. Let $i=1,2,3,4$ and $N_i$ denote the tuples $(x_i, y_i)\in (\gf(p^m)^*)^2$ satisfying the $i$-th system of equations in (\ref{system1}), respectively. Then $N_i= 4$ or $0$ for any $b \in \gf(p^m)$. Moreover, $N_i=0$ if $b=\pm 1$.
\end{lemma}
\begin{proof}
 We only show the possible values of $N_1$ and $N_2$. The possible values of $N_3$ and $N_4$ can be computed similarly.

Firstly, we consider the system $\mathbf{(I)}$ and calculate the possible values of $N_1$. There exists an element $t\in \gf(p^{2m})\backslash \gf(p^m)$ such that $t^2=-1$.
The equation $x_1^2+y_1^2=1$ can be rewritten as
\begin{equation}\label{eq:x1y1}
x_1^2-t^2y_1^2=(x_1-ty_1)(x_1+ty_1)=1.
\end{equation}
Denote $\theta=x_1-ty_1$ and $\theta^{-1}=x_1+ty_1$ in Eq. (\ref{eq:x1y1}). So, all solutions of Eq. (\ref{eq:x1y1}) can be expressed as
\begin{equation}\label{e1qdsf}
x_1=\frac{\theta+\theta^{-1}}{2}\,\, \text{and}\,\, y_1=\frac{t(\theta-\theta^{-1})}{2}.
\end{equation}
Since $x_1^{p^m}=x_1$, $y_1^{p^m}=y_1$ and $t^{p^m} = -t$, we have
$$ (\theta+\theta^{-1})^{p^m}=\theta+\theta^{-1}\,\,\text{and}\,\,(\theta-\theta^{-1})^{p^m}=-(\theta-\theta^{-1}).$$
These are equivalent to
$$ (\theta^{p^m+1}-1)(\theta^{p^m-1}-1)=0 \,\,\text{and}\,\,(\theta^{p^m+1}-1)(\theta^{p^m-1}+1)=0. $$
Hence,
\begin{equation}\label{theta}
 \theta^{p^m+1}=1.
\end{equation}
From Eq.(\ref{e1qdsf}) we obtain
\begin{equation*}
x_1^{p^k+1}-y_1^{p^k+1}=\frac{1}{4}((\theta+\theta^{-1})^{p^k+1}-(\theta-\theta^{-1})^{p^k+1})
=\frac{1}{2}(\theta^{p^k-1}+\theta^{1-p^k})=-b^{\frac{p^k+1}{2}}.
\end{equation*}
Let $\gamma =\theta^{p^k-1}$. This equation is rewritten as
\begin{equation}\label{eq:sd1fs}
\gamma^2+2b^{\frac{p^k+1}{2}}\gamma + 1 =0.
\end{equation}

Assume that $b=\pm 1$. Then $b^{\frac{p^k+1}{2}}=1$ and Eq.(\ref{eq:sd1fs}) has only one solution $\gamma=-1$. So, $\theta^{2(p^k-1)}=1$.  Since $p^m\equiv 3 \pmod 4$ and $p^k \equiv 3 \pmod 4$, we know that $m$ and $k$ are odd. By  Lemma \ref{lemma5}, $\gcd(2(p^k-1),p^m+1)= 2\gcd(p^k-1,p^m+1)=4$. From Eq.(\ref{theta}) we have $\theta^4=1$. This means that $\theta^2=\pm 1$. However, $\theta^2=1$ is contradictory to that $\theta^{p^k-1}=-1$. Hence, $\theta^2=-1$, i.e., $\theta=-\theta^{-1}$. This is impossible since $x_1\neq 0$. Therefore, $N_1=0$ if $b=\pm 1$.

Assume that $b\neq \pm1$. From Lemma \ref{lemma5}, it is easy to check that $\gcd(p^m-1,\frac{p^k+1}{2})=1$ since $p^m\equiv 3 \pmod 4$ and $p^k \equiv 3 \pmod 4$. Then one can deduce that $b^{\frac{p^k+1}{2}} \neq \pm 1$. So, Eq.(\ref{eq:sd1fs}) has no or two solutions in $\gf(p^{2m})$. If Eq.(\ref{eq:sd1fs}) has two solutions $\gamma_1$ and $\gamma_2$. From Eq.(\ref{theta}) we have
\begin{equation}\label{e1qpmd}
\gamma_1=\theta^{p^k-1} , \,\, \theta^{p^m+1}=1,
\end{equation}
and
\begin{equation}\label{e1qpmd1}
 \gamma_2=\theta^{p^k-1}, \,\, \theta^{p^m+1}=1.
\end{equation}
Since $\gamma_1\gamma_2 =1$, $\theta\in \gf(p^{2m})$ satisfies Eq.(\ref{e1qpmd}) if and only if $\theta^{-1}$ satisfies Eq.(\ref{e1qpmd1}). If $\theta_1, \theta_2\in \gf(p^{2m})$ satisfy Eq.(\ref{e1qpmd}),
then $(\frac{\theta_1}{\theta_2})^{p^m+1}=(\frac{\theta_1}{\theta_2})^{p^k-1}=1$. So, $(\frac{\theta_1}{\theta_2})^2=1$ since $\gcd(p^k-1,p^m+1)=2$. As a result, if there is a $\theta$ satisfying Eq.(\ref{e1qpmd}),
then all solutions of Eq.(\ref{e1qpmd}) can be represented as $\pm \theta$, and all solutions of Eq.(\ref{e1qpmd1}) can be represented as $\pm \theta^{-1}$. Therefore, $N_1= 4$ or $0$ for any $b \in \gf(p^m)
\setminus \{ \pm 1\}$.

Secondly, we study the system $\mathbf{(II)}$ and calculate the possible values of $N_2$. From the first equation of the system $\mathbf{(II)}$, we know
 $$ x_2^2-y_2^2=(x_2-y_2)(x_2+y_2)=1.$$
Let $\delta=x_2-y_2$ and $\delta^{-1}=x_2+y_2$. Then,
\begin{equation*}\label{eqdddd}
x_2=\frac{\delta+\delta^{-1}}{2}\,\, \text{and}\,\, y_2=\frac{\delta-\delta^{-1}}{2}.
\end{equation*}
Substituting $x_2$ and $y_2$ into the second equation of the system $\mathbf{(II)}$, we have
\begin{equation}\label{eq:s1dfs}
x_2^{p^k+1}+y_2^{p^k+1}=\frac{1}{4}((\delta+\delta^{-1})^{p^k+1}+(\delta-\delta^{-1})^{p^k+1})
=\frac{1}{2}(\delta^{p^k+1}+\delta^{-(p^k+1)})=-b^{\frac{p^k+1}{2}}.
\end{equation}
Let $\nu=\delta^{p^k+1}$. Eq.(\ref{eq:s1dfs}) can be rewritten as
\begin{equation}\label{eq:delta2}
\nu^2+2b^{\frac{p^k+1}{2}}\nu+1=0.
\end{equation}

Assume that $b=\pm 1$. Analysis similar to that in above cases above implies that $N_2=0$. Assume that $b\neq \pm1$. We know that $b^{\frac{p^k+1}{2}} \neq \pm 1$ since $\gcd(p^m-1,\frac{p^k+1}{2})=1$.
So, Eq.(\ref{eq:delta2}) has no or two solutions in $\gf(p^{2m})$. If Eq.(\ref{eq:delta2}) has two solutions $\nu_1$ and $\nu_2$. Then, we have
\begin{equation*}
\delta^{p^m-1}=1, \nu_1=\delta^{p^k+1},
\end{equation*}
and
\begin{equation*}
\delta^{p^m-1}=1, \nu_2=\delta^{p^k+1}.
\end{equation*}
By a similar analysis above, we know that $N_2= 4$ or $0$ for any $b \in \gf(p^m)\setminus \{ \pm 1\}$.
\end{proof}

\begin{lemma}\label{eq:daddf}
Let $p^k\equiv 3 \pmod 4$ and $p^m\equiv 3 \pmod 4$. For $b \in \gf(p^m)$, any two systems in (\ref{system1}) cannot have solutions in $(\gf(p^m)^*)^2$ simultaneously.
\end{lemma}
\begin{proof} We only prove that the systems $\mathbf{(I)}$ and $\mathbf{(II)}$, the systems $\mathbf{(II)}$ and $\mathbf{(III)}$ cannot have solutions simultaneously. The other cases can be similarly proved.

From Lemma \ref{lem:ddf} we know that $x_1$ and $y_1$ in $\mathbf{(I)}$ can be represented as $x_1=\frac{\theta+\theta^{-1}}{2}$ and $y_1=\frac{t(\theta-\theta^{-1})}{2}$, respectively, where
$\theta \in \gf(p^{2m})$ and $\theta^{p^m+1}=1$. From the second equation of $\mathbf{(I)}$ we have
\begin{equation}\label{eq:sdfs}
\frac{1}{2}(\theta^{p^k-1}+\theta^{1-p^k})=-b^{\frac{p^k+1}{2}}.
\end{equation}
Similarly, $x_2$ and $y_2$ in $\mathbf{(II)}$ can be expressed as $x_2=\frac{\delta+\delta^{-1}}{2}$ and $y_2=\frac{\delta-\delta_2^{-1}}{2}$, respectively, where
$\delta\in \gf(p^m)$. From the second equation of $\mathbf{(II)}$ we have
\begin{equation}\label{eq:s2dfs}
\frac{1}{2}(\delta^{p^k+1}+\delta^{-(p^k+1)})=-b^{\frac{p^k+1}{2}}.
\end{equation}
From Eqs.(\ref{eq:sdfs}) and (\ref{eq:s2dfs}), we obtain
\begin{equation}\label{eq:thetadelta}
\theta^{p^k-1}+\theta^{1-p^k}-\delta^{p^k+1}-\delta^{-(p^k+1)}=0.
\end{equation}
Multiplying the both sides of Eq.(\ref{eq:thetadelta}) by $\theta^{p^k-1}\delta^{p^k+1}$, we have
$$\theta^{2(p^k-1)}\delta^{p^k+1}+\delta^{p^k+1}-\theta^{p^k-1}\delta^{2(p^k+1)}-\theta^{p^k-1}=(\delta^{p^k+1}-\theta^{p^k-1})(1-\theta^{p^k-1}\delta^{p^k+1})=0.$$
So, $\delta^{p^k+1}=\theta^{p^k-1}$ or $\delta^{p^k+1}=\theta^{-(p^k-1)}$.
Hence, $$\delta^{(p^k+1)(p^m-1)}=\theta^{-(p^k-1)(p^m-1)}=1.$$
Since $p^k\equiv 3 \pmod 4$  and $p^m\equiv 3 \pmod 4$, one can verify that $\gcd((p^k-1)(p^m-1),p^m+1)=4$ by Lemma \ref{lemma5}. So, $\theta^4=1$, i.e.,
$\theta^2=1$ or $\theta^2=-1$. If $\theta^2=1$ then $y_1=\frac{t(\theta-\theta^{-1})}{2}=0$ and if $\theta^2=-1$ then $x_1=\frac{\theta+\theta^{-1}}{2}=0$.
This is contradictory to that $x_1, y_1 \in \gf(p^m)^*$. Hence, $\mathbf{(I)}$ and $\mathbf{(II)}$ cannot have solutions $(x, y)\in (\gf(p^m)^*){^2}$ simultaneously.

Next, we show that $\mathbf{(II)}$ and  $\mathbf{(III)}$ cannot have solutions $(x, y)\in (\gf(p^m)^*){^2}$ simultaneously. From the first equation of the system $\mathbf{(III)}$, let $\gamma=x_3-y_3$ and $-\gamma^{-1}=x_3+y_3$, where $\gamma\in \gf(p^m)$. Then,
$$ x_3=\frac{\gamma-\gamma^{-1}}{2}\,\, \text{and}\,\, y_3=-\frac{\gamma^{-1}+\gamma}{2}.$$
The second equation of the system $\mathbf{(III)}$ can be rewritten as
\begin{equation}\label{eq:s3dfs}
\frac{1}{4}((\gamma-\gamma^{-1})^{p^k+1}+(\gamma^{-1}+\gamma)^{p^k+1})=\frac{1}{2}(\gamma^{p^k+1}+\gamma^{-(p^k+1)})=b^{\frac{p^k+1}{2}}.
\end{equation}
From  Eqs. (\ref{eq:s2dfs}) and (\ref{eq:s3dfs}), we have
\begin{equation}\label{eq:deltagamma}
\delta^{p^k+1}+\delta^{-(p^k+1)}+\gamma^{p^k+1}+\gamma^{-(p^k+1)}=0.
\end{equation}
Multiplying the both sides of Eq.(\ref{eq:deltagamma}) by $(\delta\gamma)^{p^k+1}$, we have
$$(\delta^2\gamma)^{p^k+1}+\delta^{p^k+1}+\delta^{p^k+1}\gamma^{2(p^k+1)}+\gamma^{p^k+1}=(\delta^{p^k+1}+\gamma^{p^k+1})((\delta\gamma)^{p^k+1}+1)=0.$$
So, $\delta^{p^k+1}=-\gamma^{p^k+1}$ or $\delta^{p^k+1}=-\gamma^{-(p^k+1)}$. This is a contradiction since $\delta, \,\, \gamma \in \gf(p^m)$ and $-1$ is a non-square element in $\gf(p^m)$.
Hence, the systems $\mathbf{(II)}$ and $\mathbf{(III)}$ cannot have solutions simultaneously.
\end{proof}

With the above preparations, we now prove the following main result.


\begin{theorem}\label{Theorem1}
Let $p^m\equiv 3 \pmod 4$. Let $k$ and $d$ be positive integers such that $d({p^k+1})\equiv2 \pmod {p^m-1}$. If $c=-1$, then $F(x)=x^d$ is PcN over $\gf(p^m)$ if and only if $d$ is odd.
\end{theorem}

\begin{proof}
 In order to prove this theorem, we need to show the equation
 \begin{equation}\label{eqddb}
 x^d+(x+1)^d=b
 \end{equation}
has at most one solution in $\gf(p^{m})$ for any $b\in \gf(p^m)$. If $d$ is even, then $x=0$ and $x=-1$ are solutions of Eq.(\ref{eqddb}) for $b=1$.
So, $d$ is odd if $F(x)=x^d$ is a PcN function.

In the following we will prove the sufficiency. Since $d({p^k+1})\equiv2 \pmod {p^m-1}$, there exists an integer $\ell$ such that
\begin{equation}\label{eq:pkpm}
d({p^k+1})=2+\ell(p^m-1).
\end{equation}
If $p^k\equiv 1 \pmod 4$, then one can deduce that $\ell$ is even from Eq.(\ref{eq:pkpm}). So,
\begin{equation}\label{eqpk11}
d({\frac{p^k+1}{2}})\equiv1 \pmod {p^m-1}.
\end{equation}
If there exists an element $d\in \gf(p^m)$ such that $(\ref{eqpk11})$ holds, from Lemma \ref{lem-05} we know that $\gcd(\frac{p^k+1}{2}, p^m-1)=1$, i.e., $\gcd(p^k+1, p^m-1)=2$. This implies that $v_2(m)\leq v_2(k)$ by Lemma~\ref{lemma5}.
From Lemmas \ref{lem:pk+1div2} and \ref{lem:inversepcn}, we know that $F(x)=x^d$ is PcN. 

Now, we show the result for the case $p^k\equiv 3 \pmod 4$. It is clear that $x=0$ and $x=-1$ are solutions of Eq.(\ref{eqddb}) for $b=1$ and $b=-1$, respectively, since $d$ is odd.
Next, we always assume that $x$ and $x+1$ are non-zero. Let SQ and NSQ be the sets of the square and non-square elements in $\gf(p^m)$, respectively.
From Lemma \ref{lem-05}, it is clear that $\gcd(p^k+1, p^m-1)=2$ and $\gcd(d, p^m-1)=1$ since $d({p^k+1})\equiv2 \pmod {p^m-1}$ and $d$ is odd. The proof can be done in the following four cases.

\noindent {\bf Case 1:} $x, x+1 \in$ SQ. We use $\alpha_0^{p^k+1}$ and $\beta_0^{p^k+1}$ to represent $x$ and $x+1$, respectively, where $\alpha_0,$ $\beta_0 \in \gf(p^m)^*$.
So, $x^d=(\alpha_0^{p^k+1})^{d}=\alpha_0^2$ and $(x+1)^d=(\beta_0^{p^k+1})^{d}=\beta_0^2$. From Eq.(\ref{eqddb}) we have the following system of equations,
\[\left\{ \begin{array}{ll}
           {\alpha_0^2+\beta_0^2=b},\\
           {\alpha_0^{p^k+1}-\beta_0^{p^k+1}=-1}.\end{array}  \right.\]
Set $\alpha_0=b^{\frac{1}{2}}\alpha_1$ and $\beta_0=b^{\frac{1}{2}}\beta_1$, then
\begin{equation}\label{eqddb1}
\left\{ \begin{array}{ll}
           {\alpha_1^2+\beta_1^2=1},\\
           \alpha_1^{p^k+1}-\beta_1^{p^k+1}=-b^{\frac{p^k+1}{2}}.\end{array}  \right.
\end{equation}
It is clear that all pairs $(\pm \alpha_1, \pm \beta_1)$ satisfying Eq. (\ref{eqddb1}) give the same pair $(x, x+1)$, i.e., the number of pairs $(\alpha_1, \beta_1)$ satisfying Eq.(\ref{eqddb1}) is four times of the number of
$x\in \gf(p^{m})^*$ satisfying Eq.(\ref{eqddb}).

\noindent {\bf Case 2:} $x \in$ SQ and $x+1 \in$ NSQ. Since $p^m\equiv 3 \pmod 4$, $-1$ is a non-square element in $\gf(p^m)$.
We use $\alpha_2^{p^k+1}$ and $-\beta_2^{p^k+1}$ to represent $x$ and $x+1$, respectively, where $\alpha_2$, $\beta_2 \in \gf(p^m)$. So, $x^d=(\alpha_2^{p^k+1})^{d}=\alpha_2^2$ and $(x+1)^d=(-\beta_2^{p^k+1})^{d}=-\beta_2^2$. From Eq.~(\ref{eqddb}) we get the following system of equations,
\[\left\{ \begin{array}{ll}
           {\alpha_2^2-\beta_2^2=b},\\
           {\alpha_2^{p^k+1}+\beta_2^{p^k+1}=-1}.\end{array}  \right.\]
Let $\alpha_2=b^{\frac{1}{2}}\alpha_3$ and $\beta_2=b^{\frac{1}{2}}\beta_3$, then
\begin{equation}\label{eqddb2}
\left\{ \begin{array}{ll}
           {\alpha_3^2-\beta_3^2=1},\\
           {\alpha_3^{p^k+1}+\beta_3^{p^k+1}=-b^{\frac{p^k+1}{2}}}.\end{array}  \right.
\end{equation}
It is easy to see that all pairs $(\pm \alpha_3, \pm \beta_3)$ satisfying Eq.(\ref{eqddb2}) give the same pair $(x, x+1)$, i.e., the number of pairs $(\alpha_3, \beta_3)$ satisfying Eq.(\ref{eqddb2}) is four times of the number of
$x\in \gf(p^{m})^*$ satisfying Eq.(\ref{eqddb}).

\noindent {\bf Case 3:} $x \in$ NSQ and $x+1 \in$ SQ. In order to determine the number of the solutions of Eq.(\ref{eqddb}) for any $b\in \gf(p^m)$, by a similar analysis to those in Case 1 and Case 2, we need to consider the number of the solutions of the following equations,
\begin{equation}\label{eqddb3}
\left\{ \begin{array}{ll}
           {\alpha_4^2-\beta_4^2=-1},\\
           {\alpha_4^{p^k+1}+\beta_4^{p^k+1}=b^{\frac{p^k+1}{2}}}.\end{array}  \right.
\end{equation}
Moreover, the number of pairs $(\alpha_4, \beta_4)$ satisfying Eq.(\ref{eqddb3}) is four times of the number of $x\in \gf(p^{m})^*$ satisfying Eq.(\ref{eqddb}).

\noindent {\bf Case 4:} $x, x+1 \in$ NSQ. In order to determine the number of the solutions of Eq.(\ref{eqddb}) for any $b\in \gf(p^m)$, by a similar analysis to those in Case 1 and Case 2, we need to consider the number of the solutions of the following equations,
\begin{equation}\label{eqddb4}
\left\{ \begin{array}{ll}
           {\alpha_5^2+\beta_5^2=-1},\\
           {\alpha_5^{p^k+1}-\beta_5^{p^k+1}=b^{\frac{p^k+1}{2}}}\end{array}  \right.
\end{equation}
Moreover, the number of pairs $(\alpha_5, \beta_5)$ satisfying Eq.(\ref{eqddb4}) is four times of the number of $x\in \gf(p^{m})^*$ satisfying Eq.(\ref{eqddb}).
Then the desired conclusion then follows from Lemmas~\ref{lem:ddf} and~\ref{eq:daddf}.
\end{proof}

\begin{example}
Let $c=-1$ and $k=1$. If
$p=3,m=5,d=61$, or $p=7,m=3,d=43$, or $p=11,m=3,d=111$,
then $F(x)=x^d$ is PcN. These results have been verified by Magma programs.
\end{example}

In the following, we discuss the $(-1)$-differential uniformity of the monomial $x^d$ over $\gf(p^m)$ for the case~$p^m\equiv 1 \pmod 4$,
where
\begin{equation}\label{eq:pk1p1}
d({p^k+1})\equiv 2 \pmod {p^m-1}.
\end{equation}
From Lemmas \ref{lemma5} and  \ref{lem-05}, there are some $d$ such (\ref{eq:pk1p1}) holds if and only if $v_2(m)\leq v_2(k)$. Obviously, the congruence~(\ref{eq:pk1p1}) is equivalent to $d({p^k+1})=2+\ell (p^m-1)$ for some integer $\ell$. If $\ell$ is even, then $d \cdot \frac{p^k+1}{2} \equiv 1 \pmod {p^m-1}$.
In this case, by Lemmas~\ref{lem:pk+1div2} and \ref{lem:inversepcn} we know that $x^d$ is a PcN function since $v_2(m)\leq v_2(k)$. If $\ell$ is odd then $d$ satisfies that
$d \cdot \frac{p^k+1}{2} \equiv \frac{p^m+1}{2} \pmod {p^m-1}$. For $d$ in this case, $(-1)$-differential uniformity of the monomial $x^d$ is given in the following theorem.

\begin{theorem}\label{Theorem2}
Let $m$ and $k$ be positive integers with $v_2(k)=v_2(m)$. Let $p^m\equiv 1 \pmod 4$ and $d\cdot \frac{p^k+1}{2}\equiv \frac{p^m+1}{2} \pmod {p^m-1}$. Then the monomial $F(x)=x^d$
is PcN over $\gf(p^m)$, where $c=-1$.
\end{theorem}
\begin{proof}
  For any $b \in \gf(p^m)$, we need to show
\begin{equation}\label{eqdpk3}
(x+1)^d+x^d=b
\end{equation}
has at most one solution in $\gf(p^m)$. Since $d \cdot \frac{p^k+1}{2} \equiv \frac{p^m+1}{2} \pmod {p^m-1}$, we have $\gcd(p^m-1,d)\,| \, \frac{p^m+1}{2}$ by Lemma~\ref{lem-05}. It is clear that $\gcd(p^m-1,\frac{p^m+1}{2})=1$ since $p^m\equiv 1 \pmod 4$ and $\gcd(p^m-1,p^m+1)=2$.
So, $\gcd(p^m-1,d)=1$.

We first assume that $x\neq 0$ and $x\neq 1$.  If $b=0$, then Eq.(\ref{eqdpk3}) becomes $(1+1/x)^d=-1$ and it has a unique solution since $\gcd(p^m-1,d)=1$. If $b\neq 0$, then Eq.(\ref{eqdpk3}) can be rewritten as
\begin{equation}\label{eqpm41}
\frac{(x+1)^d}{b}+ \frac{x^d}{b}=1.
\end{equation}
Let $h=(p-1)/4$ if $p\equiv 1 \pmod 4$ and $h=(3p-1)/4$ if $p\equiv 3 \pmod 4$.
Let $\gamma\in \gf(p^{2m})^*$ be a solution of $x^2+\mu x+h^2=0$, where $\mu \in \gf(p^m)$. It is easy to check that $h^2 \gamma^{-1}$ is also a solution of $x^2+\mu x+h^2=0$. Then $\mu =\gamma+h^2\gamma^{-1}$.
This means that any element in $\gf(p^m)$ can be expressed by $-(\gamma+h^2\gamma^{-1})$ for some $\gamma\in \gf(p^{2m})$. Let $-x^d/b$ denote by $\gamma+h^2\gamma^{-1}+2h=(\gamma+h)^2/\gamma$ and
$1-x^d/b$ denote by $\gamma+h^2\gamma^{-1}+2h+1=(\gamma-h)^2/\gamma$, i.e.,
\begin{equation}\label{eqd1pk}
x^d=-b\cdot \frac{(\gamma+h)^2}{\gamma}\, \,\,\,{\rm and }\,\,\,\, (x+1)^d = b\cdot \frac{(\gamma-h)^2}{\gamma}.
\end{equation}
Let $\eta$ denote the quadratic characteristic of $\gf(p^{m})^*$. Raising the both sides of Eqs.(\ref{eqd1pk}) to $\frac{p^k+1}{2}$th power, we have
\begin{equation}\label{eqxetax}
x\eta(x)=x^{\frac{p^m+1}{2}} = x^{d\frac{p^k+1}{2}} =-\left(\frac{b}{\gamma}\right)^{\frac{p^k+1}{2}}(\gamma+h)^{p^k+1}
\end{equation}
since $\frac{p^k+1}{2}$ is odd, and
\begin{equation}\label{eqxetax+1}
(x+1)\eta(x+1)=(x+1)^{\frac{p^m+1}{2}} = (x+1)^{d\frac{p^k+1}{2}} =\left(\frac{b}{\gamma}\right)^{\frac{p^k+1}{2}}(\gamma-h)^{p^k+1}.
\end{equation}
Since $d \cdot \frac{p^k+1}{2} \equiv \frac{p^m+1}{2} \pmod {p^m-1}$ and $p^m\equiv 1 \pmod 4$, we know that $d$ is odd. Raising the both sides of Eqs.(\ref{eqxetax}) and (\ref{eqxetax+1}) to $d$th power, respectively
and combining Eqs.(\ref{eqd1pk}), we get
\begin{equation}\label{etaxandx+1}
\eta(x)=\left(\frac{b}{\gamma}\right)^{\frac{p^m-1}{2}}(\gamma+h)^{p^m-1} \,\,{\rm and } \,\, \eta(x+1)=\left(\frac{b}{\gamma}\right)^{\frac{p^m-1}{2}}(\gamma-h)^{p^m-1}.
\end{equation}

Case I: $\eta(x+1) = \eta(x)$. From Eqs.(\ref{etaxandx+1}) we get
\begin{equation*}
 1=\frac{\eta(x+1)}{\eta(x)} =\left( \frac{\gamma-h}{\gamma+h} \right)^{p^m-1}.
\end{equation*}
This implies that $\gamma^{p^m-1} =1$, i.e., $\gamma\in \gf(p^{m})$. Eq.(\ref{eqxetax+1}) subtracting Eq.(\ref{eqxetax}) implies that
\begin{equation}\label{eqgammacase1}
\gamma^{p^k+1}-\frac{1}{2}b^{-\frac{p^k+1}{2}}\eta(x) \gamma^{\frac{p^k+1}{2}}+h^{2}=0.
\end{equation}
Set $\theta =\gamma^{\frac{p^k+1}{2}}$. Since $\gamma\in \gf(p^{m})^*$ and $\gcd(\frac{p^k+1}{2}, p^m-1)=1$, we know that $\gamma$ corresponds $\theta$ one by one.
Then Eq.(\ref{eqgammacase1}) can be rewritten as
\begin{equation}\label{eqgammacase11}
\theta^2-\frac{1}{2}b^{-\frac{p^k+1}{2}}\eta(x)\theta + h^{2}=0.
\end{equation}
It is known that Eq.(\ref{eqgammacase11}) has most two solutions $\theta_1$ and $\theta_2$ in $\gf(p^{m})$, and they satisfy $\theta_2= h^2 \theta_1^{-1}$.
Since $\gamma$ and $\theta$ are one one corresponding, we know that Eq.(\ref{eqgammacase1}) has at most two solutions $\gamma_1$ and $\gamma_2$, and they
satisfy $\gamma_2^{\frac{p^k+1}{2}}= h^2 \gamma_1^{-\frac{p^k+1}{2}}$. This implies that $\gamma_2= h^2 \gamma_1^{-1}$ since $\gcd(\frac{p^k+1}{2}, p^m-1)=1$ and $h\in \bF_p$.
Then $\gamma_1+h^2\gamma_1^{-1}+2h=\gamma_2+h^2\gamma_2^{-1}+2h$. This means that $\gamma_1$ and $\gamma_2$ gives the same value of $x$ since $-x^d/b$ is denoted by $\gamma_i+h^2\gamma_i^{-1}+2h$ for $i=1, 2$ and $\gcd(p^m-1,d)=1$. Hence, Eq.~(\ref{eqdpk3}) has at most one solution in this case.

Case II: $\eta(x+1) = -\eta(x)$. From (\ref{etaxandx+1}) we get
\begin{equation}\label{eqgamma}
 -1=\frac{\eta(x+1)}{\eta(x)} =\left( \frac{\gamma-h}{\gamma+h} \right)^{p^m-1}.
\end{equation}
This equation implies that $\left( \frac{\gamma}{h}\right)^{p^m+1} =1$, i.e., $\frac{\gamma}{h}$ is in the subgroup of $(p^m+1)$-st roots of unity in $\gf(p^{2m})^*$,
denote it by $\mathcal{U}$. Eq. (\ref{eqxetax+1}) plus Eq.(\ref{eqxetax}) implies that
\begin{equation}\label{eqgammacase2}
\gamma^{p^k-1}-\frac{1}{2}b^{-\frac{p^k+1}{2}}h^{-1}\eta(x) \gamma^{\frac{p^k-1}{2}}+1=0.
\end{equation}
Set $\delta=\left(\frac{\gamma}{h}\right)^{\frac{p^k-1}{2}}$. Since $h\in \bF_p^*$, the above equation is equivalent to
\begin{equation}\label{eqgammacase21}
\delta^2-\frac{1}{2}b^{-\frac{p^k+1}{2}}h^{\frac{p^k-3}{2}}\eta(x) \delta +1=0.
\end{equation}
Eq.(\ref{eqgammacase21}) has at most two solutions $\delta_3$ and $\delta_4$ in $\mathcal{U}$. Since $v_2(k)=v_2(m)$, one can verify that
$\gcd(\frac{p^k-1}{2}, p^m+1) =2$. So, for each solution $\delta_i (i=3,4)$ of Eq.(\ref{eqgammacase21}), there are two corresponding solutions
$\pm \frac{\gamma_i}{h}$ of Eq.(\ref{eqgammacase2}) such that $\delta_i = \left(\pm \frac{\gamma_i}{h}\right)^{(p^k-1)/2}, \,\, i=3, 4$.
So, all possible solutions of Eq.(\ref{eqgammacase2}) in $\mathcal{U}$ are $\pm \frac{\gamma_3}{h}$ and $\pm \frac{\gamma_4}{h}$.

From Eq. (\ref{eqgammacase2}), we have $\eta(x)=-2(\gamma^{p^k}h+\gamma h)(\frac{b}{\gamma})^{\frac{p^k+1}{2}}$.
If $\frac{\gamma_i}{h}$ for $i=3, 4$, is a solution of Eq.(\ref{eqgammacase2}), substituting the values of $\eta(x)$ and $\eta(x)=-\eta(x+1)$ into Eqs.(\ref{eqxetax}) and (\ref{eqxetax+1}), respectively, we get
\begin{equation}\label{eq:xandx+1}
x=-\frac{( \gamma_i/h +h)^{p^k+1}}{2\left[ (\gamma_i/h)^{p^k}h +\gamma_i\right]}\, \,\,{\rm and } \,\,\, x+1 = -\frac{( \gamma_i/h -h)^{p^k+1}}{2\left[ (\gamma_i/h)^{p^k}h +\gamma_i\right]}.
\end{equation}
Moreover, if $-\frac{\gamma_i}{h}$ for $i=3, 4$, is also a solution of Eq.(\ref{eqgammacase2}), substituting the values of $\eta(x)$ and $\eta(x)=-\eta(x+1)$ into Eqs.(\ref{eqxetax}) and (\ref{eqxetax+1}), respectively,
we get
\begin{equation}\label{eq:xandx+12}
x=\frac{( \gamma_i/h -h)^{p^k+1}}{2\left[ (\gamma_i/h)^{p^k}h +\gamma_i\right]}\, \,\,{\rm and } \,\,\, x+1 = \frac{( \gamma_i/h +h)^{p^k+1}}{2\left[ (\gamma_i/h)^{p^k}h +\gamma_i\right]}.
\end{equation}
The pairs $(x, x+1)$ in Eqs.(\ref{eq:xandx+1}) and (\ref{eq:xandx+12}) satisfying Eq.(\ref{eqpm41}) simultaneously imply that
$b=0$. This is a contradiction. So, we can assume that all possible solutions of Eq.(\ref{eqgammacase2}) in $\mathcal{U}$ are $\frac{\gamma_3}{h}$ and $\frac{\gamma_4}{h}$.
Moreover, $\left( \gamma_3 \gamma_4/h^2\right)^{\frac{p^k-1}{2}} =1$. This implies that $\left( \gamma_3 \gamma_4/h^2\right)^{2} =1$ since $\gcd(\frac{p^k-1}{2}, p^m-1)=2$.
So, $\gamma_4= \pm h^2 \gamma_3^{-1}$. A similar analysis as above implies that $\gamma_4=  h^2 \gamma_3^{-1}$.
Then $\gamma_3+h^2\gamma_3^{-1}+2h=\gamma_4+h^2\gamma_4^{-1}+2h$. This means that $\gamma_3$ and $\gamma_4$ gives the same value of $x$ since $-x^d/b$ is denoted by $\gamma_i+h^2\gamma_i^{-1}+2h$ for $i=3, 4$ and $\gcd(p^m-1,d)=1$. Hence, Eq. (\ref{eqdpk3}) has at most one solution in this case since $\gcd(p^m-1,\frac{p^k+1}{2})=1$.

Combining the above two cases, we know that for any $b\in \gf(p^m)$, Eq.(\ref{eqdpk3}) has at most one solution in $\gf(p^m)$ if $x\neq 0$ and $x\neq -1$. Obviously, $x=0$ and $x=-1$ are solutions of Eq.(\ref{eqdpk3}) for $b=1$ and $b=-1$, respectively, since $d$ is odd. In the following, we only show that there is no other solution to Eq.(\ref{eqdpk3}) than $x=0$ for the case $b=1$. The case of $b=-1$ can be similarly proved and the details are
omitted here.

Assume that $x_0$ is a solution of $(x+1)^d+x^d=1$, where $x_0\neq 0$ and $x_0\neq -1$. If $\eta(x_0)=\eta(x_0+1)$, Eq.(\ref{eqgammacase11}) becomes
\begin{equation}\label{eqgammacaxne011}
\theta^2-\frac{1}{2}\eta(x_0)\theta + h^{2}=0.
\end{equation}
It is easy to see that Eq.(\ref{eqgammacaxne011}) has only one solution $\theta=h$ or $\theta=-h$. By the definition of $\theta$, we have $\gamma^{\frac{p^k+1}{2}}=h$ or $\gamma^{\frac{p^k+1}{2}}=-h$. Since $h \in \gf(p)$, then $(\gamma^{\frac{p^k+1}{2}})^{p-1}=1$. Hence, $\gamma \in \gf(p)$ since $\gamma^{p^m-1}=1$ and $\gcd(p^m-1, \frac{(p^k+1)(p-1)}{2})=p-1$. This means that $\gamma= h$ or $\gamma=-h$. This is contradictory to
the equations in~(\ref{eqd1pk}) since $x_0\neq 0$ and $x_0\neq -1$.

If $\eta(x_0)=-\eta(x_0+1)$, then Eq.(\ref{eqgammacase21}) becomes
\begin{equation}\label{eqgammacase211}
\delta^2-\frac{1}{2}h^{\frac{p^k-3}{2}}\eta(x) \delta +1=0.
\end{equation}
It is easy to see that $h^{\frac{p^k-3}{2}}=\pm h^{-1}$ since $h \in \gf(p)$. Then we have that Eq.(\ref{eqgammacase211}) has only one solution $\delta=1$ or $\delta=-1$. By the definition of $\delta$, we have $\left( \frac{\gamma}{h}\right)^{p^k-1} =1$.  Since $\left( \frac{\gamma}{h}\right)^{p^m+1} =1$ and $\gcd(p^m+1,p^k-1)=2$, then $\gamma= h$ or $\gamma=-h$. This is contradictory to the equations in~(\ref{eqd1pk}) since $x_0\neq 0$ and $x_0\neq -1$.
The desired conclusion then follows.
\end{proof}

\begin{example}
Let $c=-1$ and $k=1$. If
$p=5,m=5,d=3645$, or $p=13,m=3,d=157$, or $p=17,m=3,d=111$,
then $F(x)=x^d$ is PcN. These results have been verified by Magma programs.
\end{example}

\begin{remark}
It is clear that \cite[Theorem 2]{Zha2020} and \cite[Theorem 4]{Zha2020} can be seen as two special cases of Theorem \ref{Theorem1} and Theorem \ref{Theorem2} for $p=3$ and $p=5$, respectively.
\end{remark}

\begin{remark}
In references~\cite{Bartoli,Hasan}, authors have showed that the monomials $x^d$ are PcN for the following exponents:
$d=p^2-p+1$, $d=p^4+(p-2)p^2+(p-1)p+1$,  $d=(p^5+1)/(p+1)$, $d=(p-1)p^6+p^5+(p-2)p^3+(p-1)p^2+p$, $d=(p-2)p^6+(p-2)p^5+(p-1)p^4+p^3+p^2+p$ and $d=(p^7+1)/(p+1)$. It is easy to show that all $d$
listed above are special solutions of $d({p^k+1})\equiv2 \pmod {p^m-1}$ for some special $k$ and $m$. Hence, our results generalizes the results about PcN monomials in~\cite{Bartoli,Hasan}.
\end{remark}

At last, we determine the $c$-differential spectrum of a class of APcN power functions.

\begin{theorem}\label{theorem2}
Let $F(x)=x^d$ be a power function over $\gf(p^m)$, where $d={p^k+1}$, $k$ is a positive integer and $p$ is an odd prime. If $c \in  \gf(p^{\gcd(m,k)})\setminus \{1\}$, then  $F(x)$ is APcN with $c$-differential spectrum
\begin{equation}\label{eqomega}
\mathbb{S}=\left\{\omega_0=\frac{p^m-1}{2},\,\,
\omega_1=1,\,\,\omega_2=\frac{p^m-1}{2}\right\}
\end{equation}
if and only if  $v_2(m)\leq v_2(k)$. If $c \notin \gf(p^{\gcd(m,k)})$, then $F(x)$ is APcN with $c$-differential spectrum
\begin{equation}\label{eqomega2}
\mathbb{S}=\left\{\omega_0=\frac{p^{m}-p^{\frac{m}{2}}}{2},\,\,\omega_1=p^{\frac{m}{2}},\,\,\omega_2=\frac{p^{m}-p^{\frac{m}{2}}}{2}\right\}
\end{equation}
if and only if $m$ is even and $k=\frac{m}{2}$.
\end{theorem}
\begin{proof} If $c=1$, the $c$-differential uniformity of $F(x)=x^{p^k+1}$ was thoroughly analyzed in \cite{Dembowski68,Gold68}. In the following, we always assume that $c\neq 1$, and investigate solutions of
$\Delta(x)=b$ for $b \in \gf(p^m)$, where
\begin{equation*}
\begin{split}
\Delta(x)&=(x+1)^{p^k+1}-cx^{p^k+1}=(1-c)x^{p^k+1}+x^{p^k}+x+1.
\end{split}
\end{equation*}
Let $a=\frac{1}{1-c}$. The equation $\Delta(x)=b$ is equivalent to
\begin{equation}\label{eq:dasf}
x^{p^k+1}+a x^{p^k}+a x+ a (1-b)=0.
\end{equation}
Let $x=y-a$, then Eq.(\ref{eq:dasf}) becomes
\begin{equation}\label{eq:sdds}
(x-a)^{p^k+1}+a (x-a)^{p^k}+a (x-a)+ a (1-b)=x^{p^k+1}+(a-a^{p^k})x-a^2+a-a b=0.
\end{equation}
It is clear that
\begin{equation}\label{eq:sd1ds}
a=a^{p^k}\Longleftrightarrow \frac{1}{1-c}=\left(\frac{1}{1-c}\right)^{p^k}\Longleftrightarrow c=c^{p^k}\Longleftrightarrow c^{p^k-1}=1
\end{equation} for any $c\neq 1$. Then $a-a^{p^k}=0$ if and only if $c \in  \gf(p^{\gcd(m,k)})\setminus \{1\}$. The proof can be done in the following two cases.

\noindent {\bf Case 1:} $c \in  \gf(p^{\gcd(m,k)})\setminus \{1\}$. In this case,  $a-a^{p^k}=0$. Since \cite[Theorem 3]{Mesnager2020} have proved that $F(x)$ is APcN if $v_2(m)\leq v_2(k)$, we here only prove the $c$-differential spectrum of $F(x)$.

Since $a-a^{p^k}=0$, Eq.(\ref{eq:sdds}) becomes
 \begin{equation}\label{eq:sdds1}
x^{p^k+1}=a^2-a+a b.
\end{equation}
From Lemma \ref{lemma5}, we have $\gcd(p^m-1,p^k+1)=2$ since $v_2(m)\leq v_2(k)$. Then Eq.(\ref{eq:sdds1})
 has no solution, or one solution, or two solutions if $a^2-a+a b$ is a non-square element, or zero, or a square element in $\gf(p^m)$, respectively. Hence, we can obtain the $c$-differential spectrum of $F(x)$, which is given in (\ref{eqomega}).

If $F(x)$ is APcN, then Eq. (\ref{eq:sdds1}) has at most two solutions. From Lemma \ref{lemma5}, we have $v_2(m)\leq v_2(k)$.
Hence, then $F(x)$ is APcN if and only if  $v_2(m)\leq v_2(k)$.

\noindent {\bf Case 2:} $c \notin \gf(p^{\gcd(m,k)})$. In this case, we have $a-a^{p^k}\neq 0$ from (\ref{eq:sd1ds}). If $a^2-a+a b=0$, then Eq.(\ref{eq:sdds}) can be rewritten as
 $$(x^{p^k}+(a-a^{p^k}))x=0.$$ It is clear that
 $x_1=0$ and $x_2=a-a^{p^{m-k}}$ are the solutions of  the above equation. If $a^2-a+a b\neq0$, by a simple substitution of variable $x$ with $\frac{a^2+a-a b}{a-a^k}x$ and dividing $(\frac{a^2-a+a b}{a-a^k})^{p^k+1}$, then Eq. (\ref{eq:sdds}) becomes
\begin{equation*}
x^{p^k+1}+Bx-B=0,
\end{equation*}
where $B=\frac{(a-a^{p^{k}})^{p^k+1}}{(a^2-a+a b)^{p^k}}$. Obviously, $B$ runs over $\gf(p^m)^*$ if $b$ runs over $\gf(p^m)\setminus (1-a)$. From Lemma \ref{lemma4}, we know that $F(x)$ is APcN if and only if $m$ is even and $k=\frac{m}{2}$. And $F(x)$ has the $c$-differential spectrum given in (\ref{eqomega2}).
The desired conclusion then follows.
\end{proof}

\begin{remark}
When $k=0$, then $x^{p^k+1}$ becomes $x^2$. Ellingsen et al. in \cite{Ellingsen2020} proved that this function is APcN over $\gf(p^m)$ for any $c\neq 1$. It is very easy to see that the $c$-differential spectrum of $x^2$ is the given in~(\ref{eqomega}).
\end{remark}

\section{Conclusions }\label{sec:concluding}
Recently, Ellingsen et al. in \cite{Ellingsen2020} proposed a new concept called multiplicative differential, and the corresponding $c$-differential
uniformity. Then some functions with low $c$-differential uniformity have been constructed. This paper continued the research in~\cite{Bartoli,Mesnager2020,Hasan,Yan2020,Zha2020},
and mainly focused on the constructions of PcN power functions. Briefly, a necessary and sufficient condition for the Gold function being PcN was given. According to numerical experiment,
we proposed a conjecture about the possible values of $d$ for $x^d$ to be PcN over $\gf(2^m)$, where $c\in \gf(2^m)$. Second, we proved that the monomial $x^d$ over $\gf(p^m)$ was PcN, where $c=-1$ and $d$ satisfies $d({p^k+1})\equiv2 \pmod {p^m-1}$. Our theorems generalized some results on PcN power functions in \cite{Bartoli,Hasan,Zha2020}. At last, the $c$-differential spectrum of a class of APcN power functions was obtained.

\end{document}